\documentclass{article}
\usepackage{graphicx}
\usepackage{braket}
\usepackage{amsthm, amssymb, amsmath}

\usepackage{appendix}
\usepackage{bm}
\usepackage[dvipsnames]{xcolor}
\usepackage{hyperref}
\usepackage[capitalise]{cleveref}

\newcommand{\state}{\ket{n_1,\cdots,n_m}}
\newcommand{\statealpha}{\ket{\alpha}}
\newcommand{\statebeta}{\ket{\beta}}
\newcommand{\coeffalpha}{\alpha_{n_1, \cdots ,n_m}}
\newcommand{\coeffbeta}{\beta_{n_1, \cdots ,n_m}}
\newcommand{\vac}{\ket{\mathsf{vac}}}
\newcommand{\Fock}{\mathcal{F}_m^n}
\newcommand{\apow}[2]{a^{\dag #2}_{#1}}
\newcommand{\homogeneous}{n_1+\cdots+n_m=n}
\newcommand{\C}{\mathbb{C}[\alpha, \overline{\alpha}]}

\renewcommand{\a}[1]{a^{\dag}_{#1}}

\newtheorem{theorem}{Theorem}
\newtheorem{proposition}{Proposition}

\newtheorem{corollary}{Corollary}

\newtheorem{definition}{Definition}

\newtheorem{remark}{Remark}

\title{Invariant in Linear Optics}
\author{Sébastien Draux\footnote{sebastiendraux@outlook.com}\ \footnotemark[2] \footnotemark[3], Simon Perdrix \footnotemark[2], Emmanuel Jeandel\footnotemark[2], Shane Mansfield\footnotemark[3]} 
\date{\footnotemark[2] Loria, Nancy\\ \footnotemark[3] Quandela, Massy}

\begin{document}

\maketitle

\begin{abstract}
    Linear optics (LO) prohibits certain transformations. In this paper, we study the conditions for a computation to be possible in LO. We find that there are finitely many polynomials such that each of these polynomials evaluates to the same value on two photonic states if and only if there is a LO circuit transforming one of these states into the other. The proof is non-constructive, so we then focus on methods to find such polynomials. 
\end{abstract}

\section{Introduction}

Photonics is a promising candidate for scalable quantum devices \cite{Kok_Munro_Nemoto_Ralph_Dowling_Milburn_2007,AbuGhanem_2024} with a multitude of applications, including fault-tolerant quantum computation \cite{Gliniasty_Hilaire_2024, Hilaire_Dessertaine_Bourdoncle_Denys_Gliniasty_Valentí-Rojas_Mansfield_2024, Bartolucci_2023}, quantum communication \cite{LCS2023Recent}, and near-term algorithms \cite{maring2024versatile}. Despite its versatility, it is fundamentally limited by the bosonic nature of photons. Entangling operations in linear optics (LO), such as controlled gates, are only possible using probabilistic methods such as post-selection and heralding \cite{knill2001scheme,knill2002quantum,Ralph72}. The success probability of a computation decreases exponentially with the number of entangling gates needed. Hence, describing the possible computations deterministically in LO is a crucial question. We examine here the problem of whether an output photonic state can be obtained from a given entry state using LO.\\

Similar questions have recently been studied, such as characterizing the gates possible in LO \cite{Garcia-Escartin_Gimeno_Moyano-Fernández_2019} or finding how much additional resource is necessary to obtain a given 2-photon state \cite{Kieling_2008,de2024simple}. We focus here on states and we only allow LO operations, leaving aside measurement-based strategies like post-selection and heralding.\\

Our approach is based on quantities preserved when a LO circuit is applied called invariants. Such quantities have been exhibited in \cite{Parellada_Gimeno_2023}. This paper carries out a systematic study of invariants in LO. Invariant theory is a well-documented branch of mathematics (see \cite{Sturmfels_2008,Mukai_Oxbury_2003,Derksen_Kemper_2015,Dolgachev_2003}) that investigates polynomial functions that are preserved under the action of some groups. In this paper, we apply and adapt standard methods of invariant theory to the action of the unitary group on the Fock space.\\

\Cref{defnot} introduces notations for the rest of the paper and recalls the basics of LO. In \cref{invsect}, we review generalities about invariant theory. It concludes with \cref{equiv}, which establishes a characterization of the transformations allowed in LO. It states that there are finitely many polynomials to evaluate on the input and output states to check if a computation is possible. This result is not constructive, so in \cref{moliensect}, we turn our interest to Molien's series, which is a powerful combinatorics tool encapsulating all the information we need about invariants. Finally, in \cref{compsect} we compute invariants for the problem of LO. General results, tools, and methods of invariant theory are not original. However, results specific to LO such as \cref{molien1,molien2} as well as the computations and findings in \cref{compsect} are new.

\section{Definitions and notations}\label{defnot}

In the rest of the paper, $n$ and $m$ are natural numbers that represent respectively the number of photons and the number of modes, $k$ is a natural number smaller than $m$ representing a given mode, and $n_1,\cdots, n_m$ are natural numbers such that $\homogeneous$. We denote $\state$ the state with $n_k$ photons in mode $k$. Let $a_k$ and $\a{k}$ be the annihilation and creation operators in mode $k$. Let $\vac = \ket{0,\cdots,0}$ be the vacuum state, the number of modes is left implicit but will always be clear in the context. With these notations:
\begin{equation}
    \state = \frac{\apow{1}{n_1}\cdots \apow{m}{n_m}}{\sqrt{n_1!\cdots n_m!}} \vac
\end{equation}

 Let $\Fock = \mathsf{Span}(\{\state \ |\ \homogeneous\})$ be the vector space of photonic states with $n$ photons on $m$ modes called the Fock space. It has dimension $\binom{n + m - 1}{n}$. For a vector $\alpha = (\alpha_{n_1,\cdots, n_m})_{\homogeneous}$, let denote
\begin{equation}
    \ket{\alpha} = \sum_{\homogeneous} \alpha_{n_1,\cdots,n_m} \state
\end{equation}
We refer to such a state either with the ket notation $\statealpha$ or just with the vector $\alpha$ of its coefficients.\\

Let $U(m)$ be the group of $m \times m$ unitary matrices and $D(m)$ the group of diagonal $m\times m$ unitaries. A LO circuit on $m$ modes is unitary $U \in U(m)$ acting on the creation operators as:
\begin{equation}
    \a{k} \mapsto \sum_{j = 1}^{m}U_{j,k} \a{j}
\end{equation}

Let $\rho(U)$ be the unitary induced by $U$ on $\Fock$. The action of $U$ on a state $\statealpha \in \Fock$ is denoted depending on the context $U.\statealpha$, $\rho(U) \ket{\alpha}$ or simply $U.\alpha$. By definition:
\begin{equation}
    U.\state = U. \left(\frac{\apow{1}{n_1}\cdots \apow{m}{n_m}}{\sqrt{n_1!\cdots n_m!}}\vac\right)=\prod_{k=1}^{m}\frac{1}{\sqrt{n_k!}}\left(\sum_{j = 1}^{m}U_{j,k} \a{j}\right)^{n_k}\vac
\end{equation}

Let $\mathcal{O}(\alpha) = \{U.\ket{\alpha} \ | \ U\in U(m)\}$ be the orbit of $\ket{\alpha}$: all the states computable from $\alpha$. If $\beta$ is another state, either $\beta\in \mathcal{O}(\alpha)$ and $\mathcal{O}(\alpha)  = \mathcal{O}(\beta)$ or $\beta\notin \mathcal{O}(\alpha)$ and then $\mathcal{O}(\alpha) \cap \mathcal{O}(\beta) = \emptyset$.\\

With this formalism, the question we want to address is the following: find a necessary and sufficient condition to tell if given two photonic states $\statealpha,\statebeta \in \Fock$ there is a unitary $U\in U(m)$ such that $\statebeta = U.\statealpha$.

\section{Invariants}\label{invsect}

To address this question of knowing when a transformation $\ket{\alpha}\mapsto\ket{\beta}$ is possible in LO, we make use of invariant theory. Invariant theory is the study of polynomial quantities that do not change when a certain class of transformation is applied. The following results are well-known general theorems in invariant theory and applications of these to the case of the action of the unitary group on the Fock space with possibly slight adjustments. Our main source for this section is the book \cite{Derksen_Kemper_2015}.

\subsection{Definitions}

Let $\C$ be the set functions from $\Fock$ to $\mathbb{C}$ that are polynomial in the coefficients $\alpha_{n_1,\cdots, n_m}$ as well as their complex conjugates $\overline{\alpha_{n_1,\cdots, n_m}}$. Let $\C_{d,d'}$ be the set of homogeneous polynomials of degree $d$ in $\alpha$ and $d'$ in $\overline{\alpha}$, each monomial has $d$ factors of the form $\alpha_{n_1,\cdots, n_m}$ and $d'$ of the form $\overline{\alpha_{n_1,\cdots, n_m}}$. Let $\C_d = \C_{d,d}$.\\

Where regular invariants study proper polynomial functions, allowing for complex conjugation is necessary, since there is no way for a polynomial to be invariant under phase transformations without the help of complex conjugation to cancel them out. Most of the theorems of invariant theory remain valid in this context, but some of them require us to be more careful. Unless stated otherwise, by a polynomial we mean an element of $\C$.\\

\begin{definition}\label{def_inv}
    An invariant is a function $f\in \C$ such that for any unitary $U \in U(m)$ $f(U.\alpha) = f(\alpha)$.\\
    An phase invariant is a function $f\in \C$ such that for any diagonal unitary $D \in D(m)$ $f(D.\alpha) = f(\alpha)$.\\
\end{definition}

The most basic example of an invariant, and in particular of phase invariant, is the squared norm:
\begin{equation}
    \| \alpha \|^2 = \sum_{\homogeneous} n_1! \cdots n_m! |\coeffalpha|^2
\end{equation}

Phase invariants include the squared modulus of the coefficients: $|\coeffalpha|^2$ but also (for $n = m = 2$) $\alpha_{2,0}\alpha_{0,2} \overline{\alpha_{1,1}}^2$.\\

\begin{theorem}
Let $\ket{\alpha},\ket{\beta} \in \Fock$.
    \begin{itemize}
        \item If $\ket{\beta}$ can be obtained from $\ket{\alpha}$ using LO, then for all invariants $f$: $f(\beta) = f(\alpha)$
        \item If $\ket{\beta}$ can't be obtained from $\ket{\alpha}$ using LO, then there is an invariant $f$ such that $f(\beta) \neq f(\alpha)$
    \end{itemize}
\end{theorem}
The first point is obvious and the second is proved later, in the form of \cref{distinct}. In the second case we will see that the invariant $f$ witnessing the impossibility of a computation in LO can be found among a finite set.\\

As an example, with $n = m = 2$, we will see that $f(\alpha) = |\alpha_{1,1} - 4 \alpha_{2,0}\alpha_{0,2}|^2$ is an invariant, so $\statebeta = \frac{1}{\sqrt{2}}\ket{2,0}$ cannot be obtained from $\statealpha = \ket{1,1}$ since $f(\alpha) = 1 \neq f(\beta) = 0$.\\

Let's make several remarks. First, since $f(U.\alpha)$ is obtained from $f(\alpha)$ by replacing each $\coeffalpha$ with a linear combination, the degree does not change. If $f$ in an invariant and $f = \sum_{d,d'}f_{d,d'}$ where $f_{d,d'} \in \C_{d,d'}$ then $f(\alpha) = f(U.\alpha) =  \sum_{d,d'}f_{d,d'}(U.\alpha)$. By identifying the degrees, we see that each $f_{d,d'}$ is an invariant. So we can restrict our study to homogeneous invariants.\\

Second, when $U$ is diagonal and $f(\alpha)$ is any monomial, then $f(U.\alpha)$ is still the same monomial up to a phase. It follows by the identification of monomials that, to study phase invariants, it is enough to study phase-invariant monomials. As all invariants are in particular also phase invariants, if f is an invariant, then all its monomials are phase invariants.\\



Finally, it is easy to see that for phases to cancel out, phase-invariant monomials need to have the same total degree in $\alpha$ and $\overline{\alpha}$, as in $\alpha_{2,0}\alpha_{0,2} \overline{\alpha_{1,1}}^2$. From the previous remarks it follows that we can focus on invariants in $\C_d$ for all values of $d$.

\subsection{Averaging operator}

Let's now define a way to produce a lot of invariants and hopefully enough to tell states apart if a transformation is impossible.

\begin{definition}
Let $f$ be a polynomial. Define
    \begin{equation}
        f^*(\alpha) = \int_{U(m)} f(U.\alpha)dU
    \end{equation}
    where the integral is taken with respect to the Haar measure of $U(m)$ \cite{Sepanski_2007}.
\end{definition}

This operation replaces $f(\alpha)$ by its average value on $\mathcal{O}(\alpha)$. Since for $U' \in U(m)$ we have $\mathcal{O}(\alpha) = \mathcal{O}(U'.\alpha)$, it is clear that $f^*(\alpha) = f^*(U'.\alpha)$ so $f^*$ is an invariant. Formally, it corresponds to the change of variable $U \mapsto U'U$ in the integral. Note that if $f$ is invariant, $f^* = f$, hence all invariant $g$ are of the form $f^*$ (take $f = g$). Moreover, if $f$ is an invariant and $g$ is another polynomial, then $(gf)^* = g^*f$.\\

The function $f^*(\alpha)$ is still a polynomial because when evaluating this integral and expanding the expression, all the factors $\coeffalpha$ can be taken out of the integral leaving coefficients of the form 
$$\int_{U(m)} U_{i_1,j_1}\cdots U_{i_d, j_d} \overline{U_{i'_1,j'_1}}\cdots\overline{U_{i'_{d'}, j'_{d'}}} dU$$
These integrals have been calculated in \cite{Collins_Śniady_2006}.

\begin{theorem}
    \begin{align*}
        \int_{U(m)}& U_{i_1,j_1}\cdots U_{i_d, j_d} \overline{U_{i'_1,j'_1}}\cdots\overline{U_{i'_{d}, j'_{d}}} dU\\
        =& \sum_{\sigma,\tau\in S_d} \delta_{i_1, i'_{\sigma(1)}} \cdots \delta_{i_d, i'_{\sigma(d)}}\delta_{j_1, j'_{\sigma(1)}} \cdots \delta_{j_d, j'_{\sigma(d)}}W(\tau \sigma^{-1})
    \end{align*}
    where $W$ is the Weingarten function defined as:
    \begin{equation*}
        W(\sigma) = \frac{1}{d!^2}\sum_{\lambda}\frac{\chi^\lambda(1)^2\chi^\lambda(\sigma)}{s_{\lambda,n}(1)}
    \end{equation*}
    where the sum ranges over all partition $\lambda$ of $d$, $\chi^\lambda$ is the character of $S_d$ corresponding to $\lambda$ and $s_{\lambda,n}$ is the Schur polynomial of $\lambda$.\\
    If $d \neq d'$:
    $$\int_{U(m)} U_{i_1,j_1}\cdots U_{i_d, j_d} \overline{U_{i'_1,j'_1}}\cdots\overline{U_{i'_{d'}, j'_{d'}}} dU = 0$$
\end{theorem}

All the invariants can be obtained with this procedure. However, we don't need to compute $f^*$ for all $f$. Let $f$ be an invariant, write $f = \sum_j f_j$ where each $f_j$ is a monomial, hence a phase-invariant one. Then $f = f^* = \sum_j f_j^*$. So we just need to average phase-invariant monomials to compute all the invariants. Actually, averaging a polynomial that is not phase-invariant gives 0 as we will see it later.\\

As an example of a phase-invariant monomial, consider $|\coeffalpha|^2$. The following result will be justified in \cref{tenssect}.

\begin{proposition}\label{avgdeg2}
    Up to some multiplicative constant, $\left(|\coeffalpha|^2\right)^*$ is $\| \alpha \|^2$
\end{proposition}

\subsection{A characterization of possible LO computations}

In this subsection, we present a characterization of allowed computations in LO. The proof utilizes the amazing algebraic and analytical properties of polynomials. On one hand, the Hilbert basis theorem says that polynomial ideals are finitely generated, and on the other hand, the Stone-Weierstrass theorem gives that any continuous function can be approximated by polynomials. Combining these with the averaging operator, we get the best of both worlds: a lot of flexibility encapsulated in finitely many generators.

\begin{proposition}\label{finite}
    For $n$ and $m$ fixed, there is a finite set of invariants $f_1,\cdots,f_N$ such that any invariant can be expressed as a sum and product of these.
\end{proposition}

\begin{proof}
    By the Hilbert basis theorem applied to the ideal $I^+$ generated by invariants of strictly positive degree, there exists a finite family $f_1, \cdots, f_N \in I^+$ such that any invariant $f$ can be written as $f = g_1f_1+\cdots g_Nf_N$ where the $g_j$ can be any polynomial.\\
    Let's now do an induction on the degree of $f$. We have $f = f^* = g_1^*f_1+\cdots g_N^*f_N$. The $g_j^*$ have a degree strictly less than $f$ so by induction they can be expressed as a sum and product of $f_1,\cdots,f_N$ so does $f$.
\end{proof}

Hilbert famously proved his basis theorem precisely for this purpose \cite{Hilbert_1890}. However, his original proof is nonconstructive. Later, a constructive method has been found that relies on Gröbner bases. 

\begin{proposition}\label{distinct}
    If $\statebeta \neq U.\statealpha$ for all $U \in U(m)$, then there is an invariant $f$ such that $f(\alpha) \neq f(\beta)$.
\end{proposition}

\begin{proof}
    By hypothesis, $\beta \notin \mathcal{O}(\alpha)$ so $\mathcal{O}(\alpha) \cap \mathcal{O}(\beta) = \emptyset$. Moreover, since $U(m)$ is compact, so are $\mathcal{O}(\alpha)$ and $\mathcal{O}(\beta)$. By the Stone Weierstrass theorem (complex version), there is a polynomial $f$ (with conjugations) such that $|f(\alpha') - 1| < 1$ for any $\alpha'\in \mathcal{O}(\alpha)$ and $|f(\beta') + 1| < 1$ for any $\beta'\in \mathcal{O}(\beta)$. The function $f$ is not invariant, but we can replace it with $f^*$, which is an invariant and still satisfies $|f^*(\alpha) - 1| < 1$ and $|f^*(\beta) + 1| < 1$. So in particular $f^*(\alpha) \neq f^*(\beta)$
\end{proof}

Here, it is crucial for us to allow for complex conjugate, as it is a condition in the complex version of the Stone-Weierstrass approximation theorem. In regular invariant theory, this result does not hold and orbits may fail to be distinguished using invariants.

\begin{corollary}\label{equiv}
    For fixed $n$ and $m$, there is a finite set of invariant $f_1,\cdots,f_N$ such that given $\alpha,\beta \in \Fock$: $f_j(\alpha) = f_j(\beta)$ for all $j$ if and only if the computation $\statealpha \mapsto \statebeta$ is possible in LO.
\end{corollary}

\begin{proof}
    Let's take the finite family $f_1,\cdots,f_N$ given by \cref{finite}. One implication is true by definition. The other is a direct consequence of \cref{finite,distinct}.
\end{proof}

As we said before, constructing such a finite generating set of invariants is not obvious, and even using Gröbner bases, we have no guarantees on their number $N$ and their degrees. It is an active research question in invariant theory to obtain information on these generators in general \cite{Wehlau,Derksen_2001,Symonds_2011}.

\section{Molien's series}\label{moliensect}
\subsection{Generalities}
The Molien's series is a combinatorial tool to study invariants. It is the generating series of the number of invariants of each degree. As such, computing it gives information on the generators.

\begin{definition}
    Fix $n$ and $m$. Let $I \subset \C$ be the ring of invariant polynomials. We define the Molien's series of $I$ as ($n$ and $m$ are left implicit):
    $$
    F(z) = \sum_{d = 0}^{+ \infty} \dim(I \cap \C_{d,d'}) z^d \overline{z}^{d'}
    $$
    Let $F'$ be the Molien series of phase invariants.
\end{definition}

This series encodes all the information we want about invariants and their generators. The definition still makes sense when the ring $I$ of invariants is replaced by another ring. In particular, we use it for subring of invariants.
\begin{theorem}\label{frac}
    If $I$ is generated by $f_1 \in \C_{d_1,d'_1}, \cdots, f_N \in \C_{d_N,d'_N}$ then:
    $$F(z) = \frac{P(z)}{(1 - z^{d_1}\overline{z}^{d'_1}) \cdots(1 - z^{d_N}\overline{z}^{d'_N})}$$
    where $P$ is some polynomial.
\end{theorem}

The denominator contains information on the number of generators and their degree, while the polynomial $P$ in the numerator witnesses relations between them.

\begin{remark}\label{hidden}
We must be careful here. From the expression of $F$, one cannot deduce the degrees of a generating family. In particular, if we find an expression of $F$ with $P = 1$, the denominator may fail to witness all the generators. Some factors can simplify with the numerator and so won't appear in the expression of $F$.
\end{remark}

\subsection{Computation of the series}

The following theorem, known as Molien's formula \cite{Molien_1897}, is a way to compute this series. In the standard case, the proof relies on arguments from representation theory. Homogeneous polynomials of degree $d$ are seen as a representation of the group of interest, and the dimension of invariants of that degree is the number of copies of the trivial representation in this representation. This number can be computed as an average of the character, which is known in this case. For us (when allowing complex conjugation), the proof is essentially the same and can be found in \cite{Forger_1998}.
\begin{theorem}
    $$F(z) = \int_{U(m)} \frac{dU}{|\det(I - z\rho(U))|^2}$$
\end{theorem}

Finally, $\det(I - z\rho(U))$ only depends on the conjugacy class of $U$ in $U(m)$, it doesn't change when $U$ is replaced with $VUV^\dag$. Weyl integral formula \cite{Sepanski_2007} applies, and the integral reduces to one on the diagonal matrices of $U(m)$, the phases, with respect to a certain measure. More precisely, we have the following result.

\begin{theorem}
    $$F(z) = \frac{1}{m!(2 \pi)^m} \int_{0}^{2\pi}\cdots\int_{0}^{2\pi} \frac{\prod_{k < l} |e^{i\theta_k} - e^{i\theta_l}|^2}{\prod_{n_1+\cdots+n_m = n}|1 - z e^{i(n_1 \theta_1+\cdots+n_m\theta_m)}|^2} d\theta_1\cdots d \theta_m  $$
\end{theorem}

We will use this expression of $F$ to compute it in some special cases. Let 
\begin{align*}
    G(\omega_1,\cdots,\omega_m) &= \frac{\prod_{k < j} (\omega_j - \omega_k)}{\prod_{\homogeneous} (1 - z \omega_1^{n_1}\cdots\omega_m^{n_m})}\\
    &= \sum_{q_1,\cdots,q_m = 0}^{+\infty} c_{q_1,\cdots,q_m}(z)\omega_1^{q_1}\cdots\omega_m^{q_m}
\end{align*}
so that by Parseval's theorem:
\begin{equation}
    F(z) = \frac{1}{m!} \sum_{q_1,\cdots,q_m = 0}^{+\infty} |c_{q_1,\cdots,q_m}(z)|^2
\end{equation}

To compute $F$, we will develop $G$ as a series and identify the coefficients $c_{q_1,\cdots,q_m}(z)$. 
The next two results give a computation of the Molien's series in the cases $n=1$ and $n = 2$. The proofs can be found in \cref{proofmolien}.
\begin{proposition}\label{molien1}
    For $n = 1$ and any $m$:
    $$
    F(z) = \frac{1}{1 - |z|^2}
    $$
\end{proposition}

\begin{proposition} \label{molien2}
    For $n = 2$ and any $m$:
    $$
    F(z) = \frac{1}{(1 - |z|^2)(1 - |z|^4) \cdots (1 - |z|^{2m})}
    $$
\end{proposition}

When $n = 1$, there is no hidden generator. The only possibility for the Molien's series to be as computed in \cref{molien1} is for the invariants to be generated by a single polynomial of degree 2, namely $\|\alpha\|^2$. This corroborates the fact that any single-photon transformation is possible as long as it preserves the norm.\\

Considering how much more difficult the proof is for $n = 2$ compared to $n = 1$, one can imagine why the next cases are out of reach with this method. We got lucky that in these first two cases the coefficients end up being just $-1, 0$ or $1$, it is no longer the case for $n \geq 3$.\\

Similarly for phase-invariant:
\begin{proposition}\label{molien-phase}
    We have:
    \begin{align*}F'(z) &=\int_{D(m)}\frac{dU}{|1-z\rho(U)|^2}\\ 
    &= \frac{1}{(2 \pi)^m} \int_{0}^{2\pi}\cdots\int_{0}^{2\pi} \frac{d\theta_1\cdots d \theta_m }{\prod_{n_1+\cdots+n_m = n}|1 - z e^{i(n_1 \theta_1+\cdots+n_m\theta_m)}|^2}
    \end{align*}
    Let 
    \begin{align*}
        G'(\omega_1,\cdots,\omega_m) &= \frac{1}{\prod_{\homogeneous} (1 - z \omega_1^{n_1}\cdots\omega_m^{n_m})}\\
    &= \sum_{q_1,\cdots,q_m = 0}^{+\infty} c'_{q_1,\cdots,q_m}(z)\omega_1^{q_1}\cdots\omega_m^{q_m}
    \end{align*}
    Then:
    $$F'(z) = \sum_{q_1,\cdots,q_m = 0}^{+\infty} |c'_{q_1,\cdots,q_m}(z)|^2$$
\end{proposition}

\section{Applications}\label{compsect}

\subsection{Phase-invariant monomials}

Given generic states $\ket{\alpha}, \ket{\beta} \in \Fock$, the question of whether or not there is a LO circuit transforming one into the other just means finding a solution $U \in U(m)$ to the system 
\begin{equation}\label{system}
    \ket{\beta} = \rho(U)\ket{\alpha}
\end{equation}

This is a polynomial system of equations in $U$, so one can use Gröbner bases to study it. More precisely, if we add the constraints $U^\dag U = I$ to the system and use Gröbner bases algorithm to eliminate $U$ from the system, we get a generating set of relations $f(\alpha, \beta) = 0$ that need to be satisfied. However, in practice, this computation is impossible since computing Gröbner bases is too time-consuming.\\

Since we know that invariants are obtained from averaging phase-invariant monomials, we focus on the case where $U$ is diagonal. Let $U = diag(\omega_1, \cdots \omega_m)$. In this case, the system (\ref{system}) becomes
\begin{equation}\label{diagsyst}
    \coeffbeta = \omega_1^{n_1}\cdots \omega_m^{n_m} \coeffalpha
\end{equation}
for each $\homogeneous$. We also need to add new variables representing the complex conjugate $\overline{\alpha},\overline{\beta}$ and $\overline{U}$, with the equation 
\begin{equation}\label{diagconj}
        \overline{\coeffbeta} = \overline{\omega_1}^{n_1}\cdots \overline{\omega_m}^{n_m} \overline{\coeffalpha}
\end{equation}

These new variables are treated independently from their counterparts, the bar notation is just a convenient notation, but formally, they are new symbols. Finally, we also need to add equations to ensure that $U$ is unitary, namely for each $j = 1,\cdots ,m$:
\begin{equation}\label{unitsyst}
    \omega_j\overline{\omega_j} = 1
\end{equation}

To compute generating sets of phase invariants, we eliminate the variables $\omega_j$ and $\overline{\omega_j}$ from the system defined by \cref{diagsyst,diagconj,unitsyst} by computing a Gröbner basis. The wanted set is the set of monomials $f$ such that $f(\alpha) - f(\beta)$ belongs to the calculated Gröbner basis \cite{Sturmfels_2008}. Using this method, we computed generating sets of phase invariants which can be found in \cref{listmonom}.\\

For example, for $n=m=2$, we have 5 generating invariants: $f_1 = |\alpha_{20}|^2,f_2 = |\alpha_{11}|^2,f_3 = |\alpha_{02}|^2,f_4 = \alpha_{11}^2\overline{\alpha_{02}\alpha_{02}}$ and $f_5 = \overline{f_4}$. They satisfy a unique relation, namely: $f_4f_5 = f_1^2f_2f_3$ of degree 8. Hence, the Molien series is 
$$
F'(z) = \frac{1-|z|^8}{(1-|z|^2)^3(1-|z|^4)^2}
$$
which can be verified numerically using \cref{molien-phase}.
\subsection{General case}

Let's sum up our approach to determine if a state is reachable from another.
\begin{itemize}
    \item First, we compute the Molien's series. If possible, we compute it exactly otherwise we compute the first terms of its power expansion.
    \item Second, we compute invariants in the relevant degrees witnessed by the Molien's series using the averaging operator.
    \item Repeat the previous step until the Molien's series of the computed invariants matches the real Molien's series.
    \item Finally, we evaluate these invariants.
\end{itemize}

Although the first two steps are expensive, they only need to be done once for each value of $n$ and $m$.\\

As an example, for $n = m = 2$ we have $$F(z) = \frac{1}{(1 - |z|^2)(1-|z|^4)}$$ so it is clear that 2 and 4 are degrees of interest. In degree 2, we have the squared norm: $\|\alpha\|^2= 2|\alpha_{20}|^2 + 2|\alpha_{02}|^2 + |\alpha_{11}|^2$ and in degree 4, we can take:
\begin{align*}
\left(|\alpha_{20}|^4\right)^* &= \frac{8}{15} (
6 |\alpha_{02}|^{4} + 6 |\alpha_{20}|^{4} + |\alpha_{1 1}|^{4} + 6 |\alpha_{0 2}\alpha_{1 1}|^2 + 6 |\alpha_{20}\alpha_{1 1}|^2\\ &+ 4 |\alpha_{0 2}\alpha_{2 0}|^2 + 2 \alpha_{0 2}\alpha_{2 0}\overline{\alpha_{1 1}}^{2} + 2 \alpha_{1 1}^{2}\overline{\alpha_{0 2}}\overline{\alpha_{2 0}})
\end{align*}

These two invariants are algebraically independent, so their Molien's series matches $F$ and we are done. Note that the choice of $\left(|\alpha_{20}|^4\right)^*$ was arbitrary, many other polynomials would have worked. For example:
\begin{equation}
    3(\alpha_{20}\alpha_{02} \overline{\alpha_{11}}^2)^* = 3 \left(|\alpha_{20}|^4\right)^* - 2 \|\alpha\|^4
\end{equation}

\subsection{Tensor invariants}\label{tenssect}

In this subsection, we explore another way to compute invariants. It is much more convenient and also enough to provide a generating set. Instead of using the vector $\alpha$, we can represent a bosonic state by a $m \times \cdots \times m$ (with $n$ factors) symmetric tensor $A$ where the coefficient $A_{k_1, \cdots ,k_n}$ is the amplitude corresponding to the $j-$th photon being in mode $k_j$ for each $j$. A coefficient $\coeffalpha$ is split between all the possible repartition of photons giving the right number of in each mode. There are $\binom{n}{n_1, \cdots, n_m} = \frac{n!}{n_1! \cdots n_m!}$ such repartitions. So
\begin{equation}
    A_{k_1, \cdots ,k_n} = \binom{n}{n_1, \cdots, n_m}^{-1} \coeffalpha
\end{equation}
where $n_i$ is the number of indices $j$ such that $k_j = i$. Now the action of a unitary $U \in U(m)$ is easy to describe, it is just a contraction of tensors:
\begin{equation}
    (U.A)_{k_1, \cdots,k_n} = \sum_{j_1, \cdots,j_n = 1}^{m} U_{k_1,j_1}\cdots U_{k_n,j_n} A_{j_1,\cdots,j_n}
\end{equation}

With this representation, we obtain numerous invariant quantities. As an example:
\begin{equation}\label{normtens}
    \sum_{k_1, \cdots,k_n = 1}^{m} A_{k_1, \cdots,k_n} \overline{A_{k_1, \cdots,k_n}}
\end{equation}
is preserved because when a unitary is applied, it cancels out:

\begin{align*}
 &\sum_{k_1, \cdots,k_n = 1}^{m} \sum_{j_1, \cdots,j_n = 1}^{m}  \sum_{i_1, \cdots,i_n = 1}^{m} U_{k_1,j_1}\cdots U_{k_n,j_n}\overline{U_{k_1,i_1}}\cdots \overline{U_{k_n,i_n}}  A_{j_1,\cdots,j_n} \overline{A_{i_1, \cdots,i_n}}\\
 &= \sum_{j_1, \cdots,j_n = 1}^{m}  \sum_{i_1, \cdots,i_n = 1}^{m} \delta_{j_1,i_1} \cdots \delta_{j_n,i_n} A_{j_1,\cdots,j_n} \overline{A_{i_1, \cdots,i_n}}\\
 &= \sum_{j_1, \cdots,j_n = 1}^{m}  A_{j_1,\cdots,j_n} \overline{A_{j_1, \cdots,j_n}}
\end{align*}

This invariant is none other than $\|\alpha\|^2$ (up to some multiplicative constant). In general, any contraction of the following form is invariant of degree $2d$ for the same reason:
\begin{equation}\label{invtens}
    f_\sigma(\alpha) = \sum_{\bm k \in \{1, \cdots, m\}^{nd}}A^{\otimes d}_{\bm k} \overline{A}^{\otimes d}_{\sigma.\bm k}
\end{equation}
where $d$ in any natural number, $\sigma \in S_{nd}$ is any permutation, and $\sigma. \bm k$ denotes the permutation of indices: $(\sigma. \bm k)_j = \bm k_{\sigma(j)}$ for $1 \leq j \leq nd$. In other words, we pair the indices of $A^{\otimes d}$ and $\overline{A}^{\otimes d}$ two by two to contract them. This way, when a unitary is applied, it cancels out just like before. These invariants should be thought of as some sort of tensorial norm.\\ 

\begin{theorem}
    Invariants of the form $f_\sigma$ are enough to generate them all. 
\end{theorem}

\begin{proof}
    As we already know, invariants are generated by averages of monomials. Let $A^{\otimes d}_{\bm k} \overline{A}^{\otimes d}_{\bm {k'}}$ be such a monomial with $\bm k$ and $\bm{k'}$ two $nd$-tuples. Then:

\begin{align*}
    \left(A^{\otimes d}_{\bm k} \overline{A}^{\otimes d}_{\bm {k'}}\right)^* &= 
    \sum_{\bm {j,j'}\in \{1, \cdots, m\}^{nd}} A^{\otimes d}_{\bm j} \overline{A}^{\otimes d}_{\bm{j'}} \int_{U(m)} U_{k_1,j_1} \cdots U_{k_{nd},j_{nd}} \overline{U_{k'_1,j'_1}} \cdots \overline{U_{k'_{nd},j'_{nd}}} dU\\
    &= \sum_{\bm {j,j'}\in \{1, \cdots, m\}^{nd}} A^{\otimes d}_{\bm j} \overline{A}^{\otimes d}_{\bm{j'}} \sum_{\sigma, \tau \in S_{nd}} \delta_{\bm k, \tau.\bm{k'}}  \delta_{\bm j, \sigma.\bm{j'}} W(\sigma\tau^{-1})\\
    &= \sum_{\sigma,\tau \in S_{nd}} \delta_{\bm k, \tau.\bm{k'}} W(\sigma\tau^{-1}) \sum_{\bm {j}\in \{1, \cdots, m\}^{nd}}  A^{\otimes d}_{\bm j} \overline{A}^{\otimes d}_{\sigma. \bm{j}}\\
    &= \sum_{\sigma,\tau \in S_{nd}} \delta_{\bm k, \tau.\bm{k'}} W(\sigma\tau^{-1}) f_{\sigma}(\alpha)
\end{align*} 

So $\left(A^{\otimes d}_{\bm k} \overline{A}^{\otimes d}_{\bm {k'}}\right)^*$ is a linear combination of invariants of the kind $f_\sigma(\alpha)$. 
\end{proof}

Moreover, because of the term $\delta_{\bm k, \tau.\bm{k'}}$, the whole expression can only be non-zero when $\bm k$ and $\bm {k'}$ are the same up to a permutation. This is equivalent to $A^{\otimes d}_{\bm k} \overline{A}^{\otimes d}_{\bm {k'}}$ being a phase invariant.\\

Note that when $d = 1$, since $A$ is a symmetric tensor, all the invariants (\ref{invtens}) are the same regardless of $\sigma$. So for any $n$ and $m$, there is only one generating invariant of degree 2, namely $\|\alpha\|^2$ which proves \cref{avgdeg2}.\\

There is an invariant of degree $2d$ of the form (\ref{invtens}) for each $\sigma \in S_{nd}$ which gives $(nd)!$ invariant of degree $2d$. However, by symmetry, a lot of them are just the same. For example, one can permute the $d$ copies of $\overline{A}$ without changing the result. Moreover, since $A$ is a symmetric tensor, indices within the same copy of $\overline{A}$ can also be permuted.

\subsection{Two-photons states}

For two-photon states, the tensor $A$ is just a matrix and a unitary $U$ maps $A$ to $U^TAU$. The following result is a well-known consequence of Takagi's factorization \cite{Horn_Johnson_1985}.

\begin{proposition}
    A transformation $A \mapsto B$ of two-photons states is possible if and only if $A$ and $B$ have the same singular values.
\end{proposition}
We can reinterpret it and give a new proof with the spectrum of invariant theory, which exhibits a set of generating invariants. 
\begin{proof}
    The singular values of $A$ are not a polynomial function of its coefficients, however, they are the eigenvalues of the matrix $A^{\dag}A$. The characteristic polynomial of $A^{\dag}A$ does not change when a unitary is applied, so all coefficients of this polynomial are invariant. 
    \begin{equation}
        \chi_{A^\dag A}(X) = X^{m} + f_1(\alpha) X^{m-1} + f_2(\alpha) X^{m-2} + \cdots + f_m(\alpha)
    \end{equation}
    
    This way, we get an invariant of degree $2k$ for each $1 \leq k \leq m$. Moreover, these invariants are algebraically independent. This is a consequence of the fact that elementary symmetric functions are algebraically independent \cite{Macdonald_Macdonald_2015}. By \cref{molien2}, $f_1, \dots, f_m$ are a generating set of invariants so they classify the allowed transformations. The result follows immediately because the roots of $\chi_{A^\dag A}$ determine its coefficients and vice-versa.
\end{proof}
As an example, for $n = m = 2$, the two invariants are $tr(A^{\dag}A) = \|\alpha\|^2$ and $det(A^{\dag}A) = |det(A)|^2 = |\alpha_{11}-4\alpha_{20}\alpha_{02}|^2$ (up to a multiplicative constant).

\section{Conclusion}

We provided a condition which is both necessary and sufficient for computation to be possible in LO. It has the form of finitely many polynomial functions to evaluate. Although the proof is not constructive, we made exact computation for $n = 1,2$. The cases $n>2$ are still open. A continuation to this work could be to introduce post-selection and heralding to our scheme. Finally, invariants are typically a sum of many monomials so another open problem is to efficiently evaluate them on a concrete states.

\section*{Acknowledgment}
We would like to thank Boris Bourdoncle and Timothée Goubault for their feedback and constructive remarks. This work has been co-funded by the Horizon-CL4 program under the grant agreement 101135288 for EPIQUE project, by the CIFRE n° 2024/0083 and by the PROQCIMA and TUF-TOPIQC program within the French National Quantum Strategy (France 2030).

\newpage

\appendix
\section{Proofs of \cref{molien1,molien2}}\label{proofmolien}
\subsection{Proof of \cref{molien1}}

First, note that:

\begin{equation}
    \prod_{k < j} (\omega_j - \omega_k) = \sum_{\sigma \in S_m} \varepsilon(\sigma) \omega_1^{\sigma(1) - 1} \cdots \omega_m^{\sigma(m) - 1}
\end{equation}
is a Vandermonde determinant.

Let $n = 1$, in this case:
    \begin{align*}
    G(\omega_1,\cdots,\omega_m)
    &= \frac{\prod_{k < j} (\omega_j - \omega_k)}{\prod_{k = 1}^{m} (1 - z \omega_k)}\\
    &= \sum_{p_1, \cdots, p_m = 0}^{+\infty}z^{p_1+\cdots+p_m}\omega_1^{p_1}\cdots\omega_m^{p_m}\sum_{\sigma \in S_m} \varepsilon(\sigma) \omega_1^{\sigma(1) - 1}\cdots\omega_m^{\sigma(m) - 1}\\
    &= \sum_{p_1, \cdots, p_m = 0}^{+\infty}\sum_{\sigma \in S_m}\varepsilon(\sigma) z^{p_1+\cdots+p_m} \omega_1^{p_1+\sigma(1) - 1}\cdots\omega_m^{p_m+\sigma(m) - 1}\\
    &= \sum_{q_1, \cdots, q_m = 0}^{+\infty} c_{q_1,\cdots q_m}(z)\omega^{q_1}\cdots \omega^{q_m}
    \end{align*}

We need to group terms that have the same power, which means for fixed $q_1, \cdots, q_m$, studying the system:

$$\begin{cases}
    q_1 &= p_1 + \sigma(1) - 1\\
    &\vdots\\
    q_m &= p_m + \sigma(m) - 1
\end{cases}$$

Note that two solutions of this system $p = (p_1,\cdots,p_m),\sigma$ and $p' = (p'_1,\cdots,p'_m),\sigma'$ always give $p_1+\cdots+p_m =p'_1+\cdots+p'_m$, by summing all the lines of the system, so they give the same power of $z$. So we just need to add up the terms $\varepsilon(\sigma)$ for each solution, most of them will cancel out.\\

Let $p,\sigma$ be a solution and let $k$ be the smallest index such that $p_{\sigma^{-1}(k)} \neq 0$. Suppose without loss of generality, up to reordering $q$, that $\sigma = id$. Suppose that $k < m$ and let $p'_k = p_k - 1$, $p'_{k + 1} = p_{k + 1} + 1$, and for $j \neq k,k+1$: $p'_j = p_j$. Let $\sigma'$ be the transposition $(k,k+1)$. This way we have defined another solution to the system and moreover, we have that $k$ is the smallest index such that $p'_{\sigma'^{-1}(k)} \neq 0$, hence, the mapping $p,\sigma \mapsto p',\sigma'$ is an involution.
$$\begin{cases}
    q_k &= p_k + (k-1) = (p_k - 1) + k\\
    q_{k+1} &= p_{k+1} + k = (p_{k+1} + 1) + (k - 1)\\
\end{cases}$$

Also note that $\varepsilon(\sigma') = -\varepsilon(\sigma)$ so we can pair the two solutions $p,\sigma$ and $p',\sigma'$ to cancel them out in the computation of $c_{q_1,\cdots, q_m}(z)$. This way, all the solutions such that $k$ is the smallest index such that $p_{\sigma^{-1}(k)} \neq 0$ and $k < m$ will cancel out. This is valid for all $k < m$ so the only terms which have a contribution are those of the form $p = (0,\cdots,0,p_m)$ (up to permutation). They contribute to the $q$ of the form $(0, 1, \cdots, m - 2, q_m)$ with $q_m \geq m - 1$. Reciprocally, for such a $q$, there a unique $p$ of the desired form, namely $(0, \cdots,0,q_m - (m-1))$. So:
\begin{equation}
    G(\omega_1,\cdots,\omega_m) = \sum_{p = 0}^{\infty}\sum_{\sigma \in S_m} \epsilon(\sigma) z^{p} \omega_{\sigma^{-1}(1)}^{0} \omega_{\sigma^{-1}(2)}^{1}\cdots\omega_{\sigma^{-1}(m-1)}^{m-2}\omega_{\sigma^{-1}(m)}^{m-1+p}
\end{equation}
and
\begin{equation}
    F(z) = \frac{1}{m!} \sum_{p = 0}^{\infty} \sum_{\sigma \in S_m} |z|^{2p} = \frac{1}{1-|z|^2}
\end{equation}

\subsection{Proof of \cref{molien2}}

    The proof is similar to the previous one: we try to pair terms to cancel them out two by two. However, it is much more involved as the combinatorics is more complex. Let $n = 2$, in this case, we have:
    \begin{align*}
    G(\omega_1,\cdots,\omega_m)
    &= \frac{\prod_{k < j} (\omega_j - \omega_k)}{\prod_{k = 1}^{m} (1 - z \omega_k^2)\prod_{k \neq j}^{m} (1 - z \omega_k \omega_j)}\\
    &= \sum_{\substack{(p_{k,j})_{k,j}\\p_{k,j} = p_{j,k}}}\sum_{\sigma\in S_m} \varepsilon(\sigma)z^{\sum_{k, j} p_{k,j}} \omega_1^{2p_{1,1} + \sum_{k \neq 1} p_{1,k} + \sigma(1) - 1}\cdots\omega_m^{2p_{m,m} + \sum_{k \neq m} p_{m,k} + \sigma(m) - 1}
\end{align*}

Hence, the system is:

$$\begin{cases}
    q_1 &= 2p_{1,1} + p_{1,2} + \cdots + \sigma(1) - 1\\
    q_2 &= p_{2,1} + 2p_{2,2} + \cdots + \sigma(2) - 1\\
    &\vdots\\
    q_m &= p_{m,1} + \cdots + 2 p_{m,m} + \sigma(m) - 1
\end{cases}$$
with $p_{l,j} = p_{j,l}$ for all pair of indices. As before, two solutions $p,\sigma$ and $p',\sigma'$ of the same system satisfy $\sum_{k,j} p_{k,j} = \sum_{k,j} p'_{k,j}$ so they give the same power of $z$.\\

Fix $q$. Let $p,\sigma$ be a solution. As before, let $1 < k < m$ be the smallest index such that $p_{\sigma^{-1}(1),\sigma^{-1}(k)} \neq 0$. Suppose $\sigma = id$ without loss of generality. Set $p'_{k, 1} = p'_{1, k} = p_{1,k}-1$ and $p'_{k+1, 1} = p'_{1, k+1} = p_{1,k+1}+1$ and $p'_{j,l} = p_{j,l}$ for all the other pairs of indices, and $\sigma' = (k, k+1)$. The solution $p'$ also satisfies that $1<k<m$ is the smallest index such that $p'_{\sigma^{'-1}(1),\sigma^{'-1}(k)} \neq 0$. Moreover, $\varepsilon(\sigma) = - \varepsilon(\sigma')$. We pair these solutions and they cancel.
$$\begin{cases}
    q_1 &= \cdots + p_{1,k} + p_{1,k+1}+\cdots+ 0 = \cdots + (p_{1,k}-1) + (p_{1,k+1}+1) + \cdots+ 0\\

    q_k &= \cdots + p_{k, 1} + \cdots + (k - 1) = \cdots + (p_{k, 1} - 1) + \cdots + k\\
    q_{k+1} &= \cdots + p_{k+1,1} + \cdots + k = \cdots + (p_{k+1, 1} + 1) + \cdots + (k - 1)\\
\end{cases}$$
Hence, all the solutions such that $p_{\sigma^{-1}(1),\sigma^{-1}(k)} \neq 0$ for some index $1<k<m$ cancel out. Suppose now that $p_{\sigma^{-1}(1),\sigma^{-1}(k)} = 0$ for all $1<k<m$, with the same reasoning, we prove that all the solutions such that $p_{\sigma^{-1}(2),\sigma^{-1}(k)} \neq 0$ for some index $2<k<m$ don't contribute. Continuing the argument, we only need to consider solutions that are zero expect for the terms of the form $p_{\sigma^{-1}(k),\sigma^{-1}(k)}$ and $p_{\sigma^{-1}(k)\sigma^{-1}(m)}$ (for any $k$).\\

Let's now turn our interest to the later terms. Let $p, \sigma$ be a solution such that the only non-zero terms can only be those described above. We use the same trick as before, but this time using $q_m$ as an intermediate in place of $q_1$. Let $k < m - 1$ be the smallest index such that $p_{\sigma^{-1}(k),\sigma^{-1}(m)} \neq 0$. Suppose as always that $\sigma = id$. We pair this solution with $p',\sigma'$ defined as $p'_{k, m} = p'_{m, k} = p_{m,k}-1$ and $p'_{k+1, m} = p'_{m, k+1} = p_{m,k+1}+1$ and $p'_{j,l} = p_{j,l}$ for all the other pairs of indices, and $\sigma' = (k, k+1)$.

$$\begin{cases}
    q_k &= \cdots + p_{k, m}+ \cdots + (k - 1) = \cdots + (p_{k, m} - 1) + \cdots +  k\\
    q_{k+1} &= \cdots + p_{k+1,m} + \cdots + k = \cdots + (p_{k+1, m} + 1) + \cdots +  (k - 1)\\
    q_m &= \cdots + p_{m,k}+ p_{m,k+1} + \cdots = \cdots + (p_{m,k}-1) + (p_{m,k+1}+1) + \cdots\\
\end{cases}$$

As a consequence, the only remaining solutions are those with non-zero terms only among $p_{\sigma^{-1}(k),\sigma^{-1}(k)}$ (for any $k$) and $p_{\sigma^{-1}(m-1),\sigma^{-1}(m)}$. We now take care of the last non-diagonal term. Let's group these into three categories: 
\begin{itemize}
    \item $p_{\sigma^{-1}(m),\sigma^{-1}(m)} < p_{\sigma^{-1}(m - 1),\sigma^{-1}(m - 1)}$
    \item $p_{\sigma^{-1}(m),\sigma^{-1}(m)} \geq p_{\sigma^{-1}(m - 1),\sigma^{-1}(m - 1)}$ and $p_{\sigma^{-1}(m),\sigma^{-1}(m-1)} \neq 0$
    \item $p_{\sigma^{-1}(m),\sigma^{-1}(m)} \geq p_{\sigma^{-1}(m - 1),\sigma^{-1}(m - 1)}$ and $p_{\sigma^{-1}(m),\sigma^{-1}(m-1)} = 0$
\end{itemize}

Let $p,\sigma$ be of the first kind, with $\sigma = id$. Let $p', \sigma'$ defined as $p'_{m,m} = p_{m,m}$, $p'_{m-1, m-1} = p_{m-1, m-1} - 1$, $p'_{m, m-1} = p'_{m-1, m} = p_{m, m-1} + 1$ and $p'_{k,k} = p_{k,k}$ for $k\leq m-2$. Let $\sigma' = (m-1,m)$. Then $p',\sigma'$ can be paired with $p,\sigma$. Note that $p',\sigma'$ is of the second type. This time if we apply the same construction to $p',\sigma'$, we don't get $p,\sigma$ back, it is not involutive but still defines a bijection from the first kind to the second.

$$\begin{cases}
    q_{m-1} &= 2p_{m-1,m-1} + p_{m-1,m} + m - 2 = 2(p_{m-1,m-1} - 1) + (p_{m-1,m} + 1) + m - 1\\
    q_m &= 2p_{m,m} + p_{m,m-1} + m - 1 = 2p_{m,m} + (p_{m-1,m} + 1) + m - 2\\
\end{cases}$$

Finally, let $k \leq m-2$ be the largest index such that $p_{k,k} > p_{k+1,k+1}$ ($\sigma = id$). Let's define $p' = p$ except for $p'_{k,k} = p_{k,k} - 1$ and $p'_{k+2,k+2} = p_{k+2,k+2} + 1$ with $\sigma' = (k, k+2)$. Then $p,\sigma$ can be paired with $p',\sigma'$.

$$\begin{cases}
    q_{k} &= 2p_{k, k}+k-1 = 2(p_{k, k} - 1)+k+1\\
    q_{k+1} &= 2p_{k+1, k+1}+k = 2p_{k+1, k+1}+k\\
    q_{k+2} &= 2p_{k+2, k+2} + k+1 = 2(p_{k+2, k+2}+1) + k-1
\end{cases}$$

The only contributions left are the $p,\sigma$ such that $p_{k,j} = 0$ for $k\neq j$ and $p_{\sigma^{-1}(1), \sigma^{-1}(1)} \leq p_{\sigma^{-1}(2), \sigma^{-1}(2)} \leq \cdots \leq p_{\sigma^{-1}(m), \sigma^{-1}(m)}$. Back to $G$:

\begin{equation}
    G(\omega_1,\cdots,\omega_m) = \sum_{\sigma \in S_m}\sum_{p_{\sigma^{-1}(1)} \leq \cdots \leq p_{\sigma^{-1}(m)}} \varepsilon(\sigma) z^{p_1+\cdots+p_m} \omega_1^{p_1+\sigma(1)-1}\cdots\omega_m^{p_m + \sigma(m) - 1}
\end{equation}
Moreover, in this expression, no two terms have the same power of $\omega_1,\cdots,\omega_m$, so:

\begin{equation}
F(z) = \sum_{p_1 \leq \cdots \leq p_m} |z|^{2(p_1+\cdots+p_m)} = \sum_{d = 0}^{+\infty} \# \{p_1 \leq \cdots \leq p_m |p_1+\cdots+p_m = d\} |z|^{2d}
\end{equation}

It is well-known, by transposing the Young table, that 
\begin{equation}
    \# \{p_1 \leq \cdots \leq p_m |p_1+\cdots+p_m = d\} = \# (\text{partition of }d\text{ with parts } \leq m)
\end{equation}
So:

\begin{equation}
    F(z) = \sum_{d = 0}^{+\infty} \# (\text{partition of }d \text{ with parts } \leq m) |z|^{2d} = \frac{1}{(1 - |z|^2)\cdots(1-|z|^{2m})}
\end{equation}

\section{Phase-invariant monomials}\label{listmonom}
\subsection{$n = m = 2$}
There are 5 of them:
\begin{itemize}
        \item degree 2: $|\alpha_{20}|^2,|\alpha_{11}|^2, |\alpha_{02}|^2$
        \item degree 4: $\alpha_{11}^2 \overline{\alpha_{20}}\ \overline{\alpha_{02}}$ and its conjugate
    \end{itemize}
    
\subsection{$n = 2,m = 3$}
There are 26 of them:
\begin{itemize}
    \item degree 2: $|\alpha_{200}|^2, |\alpha_{110}|^2, |\alpha_{101}|^2, |\alpha_{020}|^2, |\alpha_{011}|^2, |\alpha_{002}|^2$
    \item degree 4: $\alpha_{110}^2\overline{\alpha_{020}}\ \overline{\alpha_{200}}$,
    $\alpha_{101}\alpha_{110}\overline{\alpha_{011}}\ \overline{\alpha_{200}}$,
    $\alpha_{101}^2\overline{\alpha_{002}}\ \overline{\alpha_{200}}$,
    $\alpha_{020}\alpha_{101}\overline{\alpha_{011}}\ \overline{\alpha_{110}}$,
    $\alpha_{011}\alpha_{101}\overline{\alpha_{002}}\ \overline{\alpha_{110}}$,
    $\alpha_{011}^2\overline{\alpha_{002}}\ \overline{\alpha_{020}}$ and their conjugate
    \item degree 6: $\alpha_{020}\alpha_{101}^2\overline{\alpha_{011}}^2\overline{\alpha_{200}}$,
    $\alpha_{020}\alpha_{101}^2\overline{\alpha_{002}}\ \overline{\alpha_{110}}^2$,
    $\alpha_{011}\alpha_{101}\alpha_{110}\overline{\alpha_{002}}\ \overline{\alpha_{020}}\ \overline{\alpha_{200}}$,
    $\alpha_{011}^2\alpha_{200}\overline{\alpha_{002}}\ \overline{\alpha_{110}}^2$ and their conjugate
    \end{itemize}
    
\subsection{$n = 3,m = 2$}
There are 14 of them: 
    \begin{itemize}
        \item deg 2: $|\alpha_{30}|^2,|\alpha_{21}|^2,|\alpha_{12}|^2,|\alpha_{03}|^2$
        \item deg 4: $\alpha_{21}^2\overline{\alpha_{12}}\ \overline{\alpha_{30}},\alpha_{12}\alpha_{21}\overline{\alpha_{03}}\ \overline{\alpha_{30}},\alpha_{12}^2\overline{\alpha_{03}}\ \overline{\alpha_{21}}$ and their conjugate
        \item deg 6: $\alpha_{21}^3\overline{\alpha_{03}} \overline{\alpha_{30}}^2,\alpha_{12}^3\overline{\alpha_{30}} \overline{\alpha_{03}}^2$ and their conjugate
\end{itemize}

\subsection{$n = 4,m = 2$}
There are 37 of them:
\begin{itemize}
    \item degree 2: $|\alpha_{40}|^2, |\alpha_{31}|^2, |\alpha_{22}|^2,|\alpha_{13}|^2, |\alpha_{04}|^2$

    \item degree 4: $\alpha_{31}^2\overline{\alpha_{22}}\ \overline{\alpha_{40}}$,
  $\alpha_{22}\alpha_{40}\overline{\alpha_{31}}^2$,
  $\alpha_{22}\alpha_{31}\overline{\alpha_{13}}\ \overline{\alpha_{40}}$,
  $\alpha_{22}^2\overline{\alpha_{13}}\ \overline{\alpha_{31}}$,
  $\alpha_{22}^2\overline{\alpha_{04}}\ \overline{\alpha_{40}}$,
  $\alpha_{13}\alpha_{40}\overline{\alpha_{22}}\ \overline{\alpha_{31}}$,
  $\alpha_{13}\alpha_{31}\overline{\alpha_{22}}^2$,
  $\alpha_{13}\alpha_{31}\overline{\alpha_{04}}\ \overline{\alpha_{40}}$,
  $\alpha_{13}\alpha_{22}\overline{\alpha_{04}}\ \overline{\alpha_{31}}$,
  $\alpha_{13}^2\overline{\alpha_{04}}\ \overline{\alpha_{22}}$,
  $\alpha_{04}\alpha_{40}\overline{\alpha_{22}}^2$,
  $\alpha_{04}\alpha_{40}\overline{\alpha_{13}}\ \overline{\alpha_{31}}$,
  $\alpha_{04}\alpha_{31}\overline{\alpha_{13}}\ \overline{\alpha_{22}}$,
  $\alpha_{04}\alpha_{22}\overline{\alpha_{13}}^2$ and their complex conjugate.
    \item degree 6:
    $\alpha_{31}^3\overline{\alpha_{13}}\ \overline{\alpha_{40}}^2$,
  $\alpha_{22}\alpha_{31}^2\overline{\alpha_{04}}\ \overline{\alpha_{40}}^2$,
  $\alpha_{22}^3\overline{\alpha_{13}}^2\overline{\alpha_{40}}$,
  $\alpha_{22}^3\overline{\alpha_{04}}\ \overline{\alpha_{31}}^2$,
  $\alpha_{13}\alpha_{40}^2\overline{\alpha_{31}}^3$,
  $\alpha_{13}^2\alpha_{40}\overline{\alpha_{22}}^3$,
  $\alpha_{13}^2\alpha_{40}\overline{\alpha_{04}}\ \overline{\alpha_{31}}^2$,
  $\alpha_{13}^2\alpha_{22}\overline{\alpha_{04}}^2\overline{\alpha_{40}}$,
  $\alpha_{13}^3\overline{\alpha_{04}}^2\overline{\alpha_{31}}$,
  $\alpha_{04}\alpha_{40}^2\overline{\alpha_{22}}\ \overline{\alpha_{31}}^2$,
  $\alpha_{04}\alpha_{31}^2\overline{\alpha_{22}}^3$,
  $\alpha_{04}\alpha_{31}^2\overline{\alpha_{13}}^2\overline{\alpha_{40}}$,
  $\alpha_{04}^2\alpha_{40}\overline{\alpha_{13}}^2\overline{\alpha_{22}}$,
  $\alpha_{04}^2\alpha_{31}\overline{\alpha_{13}}^3$ and their complex conjugate.
    \item degree 8:
    $\alpha_{31}^4\overline{\alpha_{04}}\ \overline{\alpha_{40}}^3$,
  $\alpha_{13}^4\overline{\alpha_{04}}^3\overline{\alpha_{40}}$,
  $\alpha_{04}\alpha_{40}^3\overline{\alpha_{31}}^4$,
  $\alpha_{04}^3\alpha_{40}\overline{\alpha_{13}}^4$ and their complex conjugate.
\end{itemize}


\bibliographystyle{IEEEtran}
\bibliography{biblio}

\end{document}